\documentclass[a4paper]{easychair}

\usepackage{microtype}

\usepackage{graphicx}

\usepackage{cite} 

\newcommand{\logicize}[1]{{\textbf{#1}}}
\renewcommand{\k}{{\logicize{K}}}
\renewcommand{\d}{{\logicize{D}}}
\renewcommand{\t}{{\logicize{T}}}
\newcommand{\kf}{{\logicize{K4}}}
\newcommand{\df}{{\logicize{D4}}}
\newcommand{\sr}{{\logicize{S4}}}
\newcommand{\sv}{{\logicize{S5}}}
\newcommand{\kv}{{\logicize{K5}}}
\newcommand{\kfv}{{\logicize{K45}}}

\usepackage{complexity}

\renewcommand{\M}{\mathcal{M}}

\usepackage{amsmath}
\usepackage{amsthm}
\usepackage{amsfonts}
\usepackage{amssymb}

\newtheorem{theorem}{Theorem}
\newtheorem{definition}{Definition}
\newtheorem{corollary}[theorem]{Corollary}
\newtheorem{lemma}[theorem]{Lemma}

\newtheorem{proposition}[theorem]{Proposition}

\renewcommand{\R}{{\mathcal{R}}}

\newcommand{\ML}{Modal Logic}
\newcommand{\NI}{Negative Introspection}

\newcommand{\constant}{|\varphi| + 2}

\title{The Completeness Problem for \ML\footnote{This research was partly supported by the project ``TheoFoMon: Theoretical Foundations for Monitorability'' (grant number: 163406-051) of the Icelandic Research Fund.}}
\titlerunning{The Completeness Problem for \ML}
\author{Antonis Achilleos}
\authorrunning{A. Achilleos}
	\institute{School of Computer Science,
		Reykjavik University \\ \texttt{antonios@ru.is}}

\begin{document}

	\maketitle
	
\begin{abstract} 
	We introduce the completeness problem for Modal Logic and examine its complexity. 
	For a definition of completeness for formulas, given a formula of a modal logic, the completeness problem asks whether the formula is complete for that logic.  
	We discover that completeness and validity have the same complexity --- with certain exceptions for which there are, in general, no complete formulas.
	To prove upper bounds, we present a non-deterministic polynomial-time procedure with an oracle from PSPACE that combines tableaux and a test for bisimulation, and determines whether a formula is complete. 
\end{abstract}

\section{Introduction}

For a modal logic $l$, we call a modal formula $\varphi$  \emph{complete} when for every modal formula $\psi$ on the same propositional variables as $\varphi$, we can derive from $\varphi$  in $l$ either the formula  $\psi$ or its negation.
For different modal logics $l$, we examine the following problem: given a modal formula $\varphi$, 
is it complete for $l$? 
We call this the completeness problem for $l$ 
and we examine its complexity. Our main results show that the completeness problem  has the same complexity as provability, at least for the logics we consider.


\ML\ is a very well-known family of logics. When one uses it to formally describe a situation, it may be of importance to be able to determine whether the formula/finite theory one uses as a description  formalizes exactly one setting (i.e. it is complete), or it leaves room for several instances consistent with this description. Given \ML's wide area of applications and the importance of logical completeness in general, we find it surprising that, to the best of our knowledge, the completeness problem for Modal Logic has not been studied as a computational problem so far.
On  the other hand, the complexity of satisfiability (and thus validity) for \ML\ has been studied extensively --- for example, see  \cite{ladnermodcomp,Halpern2007Characterizing,Halpern1992}.

We examine the completeness problem for several well-known modal logics, namely the extensions of \k\ by the axioms Factivity, Consistency, Positive Introspection, and Negative Introspection (also known as $T$, $D$, $4$, and $5$, respectively) --- i.e. the ones between \k\ and \sv. We discover that the complexity of provability and completeness tend to be the same: the completeness problem is \PSPACE-complete if the logic does not have Negative Introspection and it is \coNP-complete otherwise.
There are exceptions: for certain logics (\d\ and \t), the completeness problem as we define it is trivial, as these logics have no finite complete theories.

Our motivation partly comes from \cite{ArtemovSEL} (see also \cite{artemov2016syntactic}), where Artemov raises the following issue. It is the usual practice in Game Theory (and Epistemic Game Theory) to reason about a game based on a model of the game description. On the other hand, it is often the case in an epistemic setting that the game specification is not complete, thus any conclusions reached by examining any single model are precarious. He thus argues  for the need to verify the completeness of game descriptions, and proposes a syntactic, proof-centered approach, which is more robust and general, and which is based on a syntactic formal description of the game specification. 
Artemov's approach is more sound, in that it allows one to draw only conclusions that can be safely derived from the game specification, 
but on the other hand, the model-based approach has been largely successful in 
Game Theory for a long time. 
He explain that if we can determine that the syntactic specification of a game is complete, then the syntactic and semantic approaches are equivalent and we can describe the game efficiently, using one model.
Furthermore, he presents a complete and an incomplete formulation of the Muddy Children puzzle.

For a formula--specification $\varphi$ (for example, a syntactic description of a game), if we are interested in the formulas  we can derive from $\varphi$ (the conclusions we can draw from the game description), knowing that $\varphi$ is complete can give a significant computational advantage.
If $\varphi$ is complete and consistent, for a model $\M$ for $\varphi$, $\psi$ can be derived from $\varphi$ exactly when $\psi$ is satisfied in $\M$ (at the same state as $\varphi$). Thus, knowing that $\varphi$ is complete effectively allows us to reduce a derivability problem to a model checking problem, which is easier to solve (see, for example, \cite{Halpern1992}). 
This approach may be  useful 
when we need to examine multiple conclusions, especially if the model for $\varphi$ happens to be fairly small. On the other hand, if $\varphi$ is discovered to be incomplete, then, as a specification it may need to be refined.
We can make similar claims for other areas where Modal Logic is used as a specification language.

Notions similar to complete formulas have been studied before. 
Characteristic formulas allow one to characterize a state's equivalence class for a certain equivalence relation.
In our case, the equivalence relation is bisimulation on states of (finite) Kripke models and the notions of characteristic and complete formulas collapse, by the Hennessy-Milner Theorem \cite{hennessy1985algebraic}, in that a formula is complete for 
one of the logics we consider if and only if it is characteristic for a state in a model for that logic.
A construction of characteristic formulas
for variants of CCS processes \cite{Milner:Concurrency} was introduced in \cite{Graf_1986}. This construction allows one to verify that two CCS processes are equivalent by reducing this problem to model checking. Similar constructions were studied later in  \cite{steffen1994characteristic,M_ller_Olm_1998} for instance and in a more general manner in \cite{Aceto2015,aceto2012characteristic}.

Normal forms for Modal Logic were introduced by Fine \cite{fine1975normal} 
and they can be used to prove soundness, completeness, and the finite frame property for several modal logics with respect to their classes of frames. Normal forms are modal formulas that completely describe the behavior of a Kripke model up to a certain distance from a state, with respect to a certain number of propositional variables. Therefore, every complete formula is equivalent to a normal form, but not all normal forms are complete, as they may be agnostic with respect to states located further away. We may define that a formula is complete up to depth $d$ for logic $l$ when it is equivalent to a normal form of modal depth (the nesting depth of a formula's modalities) at most $d$.
We discuss these topics more in Section \ref{sec:conclusions}.

We focus on a definition of completeness that emphasizes on the formula's ability to either affirm or reject every possible conclusion.
We can also consider a version of the problem that asks to determine if a formula is complete up to its modal depth --- that is, whether it is equivalent to a normal form. If we are interested in completely describing a setting, the definition we use for completeness is more appropriate. However, it is not hard to imagine situations where this variation of completeness is the notion that fits better, either as an approximation on the epistemic depth agents reason with, or, perhaps, as a description of process behavior for a limited amount of time.  We briefly examine this variation in Section \ref{sec:conclusions}.



%

%

\paragraph*{Overview} 
Section \ref{sec:back}  provides background on \ML, bisimulation, and relevant complexity results. 
In Section \ref{sec:NI}, we draw our first conclusions about the completeness problem in relation to bisimulation and give our first complexity result for logics with \NI.
In Section \ref{sec:completeness}, we examine different logics and in which cases for each of these logics the completeness problem is non-trivial.
In Section \ref{sec:complexity}, we examine the complexity of the completeness problem. We first present a general lower bound. For logics with Negative Introspection we prove \coNP-completeness. For the remaining logics --- the ones without Negative Introspection for which the problem is not trivial --- we present a non-deterministic polynomial-time procedure with an oracle from \PSPACE\ that accepts incomplete formulas, as the section's main theorem, Theorem \ref{thm:main} demonstrates.
This proves that the completeness problem for these cases is \PSPACE-complete.
These complexity results are summarized in Table \ref{table:tableoftriviality}.
Section \ref{sec:proof} presents the proof of Theorem \ref{thm:main} from Section \ref{sec:complexity}.
 In Section \ref{sec:conclusions}, we consider variations of the problem and draw further conclusions.

\section{Background}
\label{sec:back}

We present needed background on \ML, its complexity, and bisimulation, and we introduce the completeness problem.
For an overview of Modal Logic and its complexity, we refer the reader 
to
\cite{MLBlackburnRijkeVenema,chagrov1997modal,Halpern1992}. We do not provide background on Computational Complexity, but the reader can see \cite{Papadimitriou1994}.

\subsection{\ML}

We assume a countably infinite set of propositional variables $p_1, p_2,\ldots$. Literals are all $p$ and $\neg p$, where $p$ is a propositional variable.
Modal formulas are constructed from literals, 
the constants $\bot, \top$, the usual operators for conjunction and disjunction $\land, \lor$ of propositional logic, and the dual modal operators, $\Box$ and $\Diamond$: 
$$\varphi :: = \bot \mid \top \mid  p \mid \neg p \mid \varphi \wedge \varphi \mid \varphi \vee \varphi \mid  \Box \varphi \mid  \Diamond \varphi .$$
The negation $\neg \varphi$ of a modal formula, implication $\varphi \rightarrow \psi$, and $\varphi \leftrightarrow \psi$ are constructed 
as usual.
The language described by the grammar above is called $L$.

For a finite set of propositional variables $P$, $L(P)\subseteq L$ is the set of formulas that use only variables from $P$.
For a formula $\varphi$, $P(\varphi)$ is the set of propositional variables that appear in $\varphi$, so $\varphi \in L(P(\varphi))$.
If $\varphi \in L$, then $sub(\varphi)$ is the set of subformulas of $\varphi$ and $\overline{sub}(\varphi) = sub(\varphi) \cup \{\neg \psi \mid \psi \in sub(\varphi) \}$.
For $\Phi$ a nonempty finite subset of $L$, $\bigwedge \Phi$ is 
a conjunction of all elements of $\Phi$ and
$\bigwedge \emptyset = \top$; we define $\bigvee \Phi$ similarly.
The modal depth $md(\varphi)$ of $\varphi$ is the largest nesting depth of its modal operators:
$md(p)=md(\bot)=0$; $md(\varphi \wedge \psi) = \max \{md(\varphi),md(\psi) \}$; and $md(\Box \varphi) = md(\varphi)+1$.
The size of $\varphi$ is $|\varphi| = |sub(\varphi)|$.
For every $d\geq 0$, $\overline{sub}_d(\varphi) = \{ \psi \in \overline{sub}_d(\varphi) \mid md(\psi) \leq d \}$.

Normal modal logics use all propositional tautologies 
and  axiom $K$,  Modus Ponens, and the Necessitation Rule:
\begin{align*}	
	&{K:\Box \varphi \wedge \Box (\varphi \rightarrow \psi) \rightarrow \Box \psi;}
	&&{\frac{\varphi \ \ \varphi \rightarrow \psi}{\psi};}
	&&{\frac{\varphi}{\Box \varphi}.}
\end{align*}
%
The logic that has \emph{exactly} these axioms and rules is the smallest normal modal logic, \k.
We can extend \k\ with more axioms: 
\begin{align*}
	&{D: \Diamond \top;}
	&&{T: \Box \varphi \rightarrow \varphi;}
	&&{4: \Box \varphi \rightarrow \Box\Box \varphi;}
	&&{5: \Diamond\varphi \rightarrow\Box\Diamond\varphi.}
\end{align*}

We consider modal logics that are formed from a combination of these axioms. Of course, not all combinations make sense: axiom $D$ (also called the Consistency axiom) is a special case of $T$ (the Factivity axiom). Axiom $4$ is called Positive Introspection and $5$ is called Negative Introspection.
Given a logic $l$ and axiom $a$, $l+a$ is the logic that has as axioms all the axioms of $l$ and $a$.
Logic \d\ is $\k + D$, \t\ is $\k+T$, $\kf = \k + 4$, $\df = \k + D + 4 = \d + 4$, $\sr = \k + T +4 = \t + 4 = \kf +T$, $\text{\logicize{KD45}} = \df+5$, and $\sv = \sr + 5$. 
From now on, unless we explicitly say otherwise, by a logic or a modal logic, we mean one of the logics we defined above.
We use $\vdash_l \varphi$ to mean that $\varphi$ can be derived from the axioms and rules of $l$; when $l$ is clear from te context, we may drop the subscript and just  write $\vdash$.

A Kripke model is a triple $\M = (W,R,V)$, where
$W$ is a nonempty set of states (or worlds), $R\subseteq W\times W$ is an accessibility
relation and $V$ is a function that assigns to each state in $W$ a set of
propositional variables. If $P$ is a set of propositional variables, then for every $a \in W$, $V_P(a) = V(a)\cap P$. To ease notation, when $(s,t)\in R$ we usually write $s R t$. 

Truth in a Kripke model is defined through relation $\models$ in the following way:
\begin{itemize}
	\item[]
	$\M,a \not \models \bot$ \ and \ $\M,a \models \top$;
	\item[]
	$\M,a \models p$ iff $p\in V(a)$ \ and \ 
	$\M,a \models \neg p$ iff $p\notin V(a)$;
	\item[]
	$\M,a \models \varphi \wedge \psi$ iff both $\M,a \models \varphi$ and $\M, a \models \psi$;
	\item[]
	$\M,a \models \varphi \vee \psi$ iff  $\M,a \models \varphi$ or $\M, a \models \psi$;
	\item[]
	$\M,a \models \Diamond \varphi$ iff there is some $b \in W$ such that $a R b$ and $\M,b \models \varphi$; and
	\item[]
	$\M,a \models \Box \varphi $ iff for all $b \in W$ such that $a R b$ it is the case that $\M,b \models \varphi$.
\end{itemize}
If $\M,a \models \varphi$, we say that $\varphi$ is true/satisfied in $a$ of $\M$. 
$(W,R)$ is called a \emph{frame}. We call a Kripke model $(W,R,V)$ (resp. frame $(W,R)$) finite if $W$ is finite.\footnote{According to our definition, for a finite model $\M=(W,R,V)$ and $a \in W$, $V(a)$ can be infinite. However, we are mainly interested in $(W,R,V_P)$ for finite sets of propositions $P$, which justifies calling $\M$ finite.} 
If $\M$ is a model (for logic $l$) and $a$ is a state of $\M$, then $(\M,a)$ is a pointed model (resp. for $l$).
For a state $x \in W$ of a frame $(W,R)$, $Reach(x)\subseteq W$ is the set of states reachable from $x$; i.e. it is the smallest set such that $x \in Reach(x)$ and if $y \in Reach(x)$ and $y R z$, 
then $z \in Reach(x)$.

Each modal logic $l$ is associated with a class of frames $F(l)$, that includes all frames $(W,R)$ for which $R$ meets certain conditions, depending on the logic's axioms.
If $l$ has axiom:
\begin{description}
	\item [$D$,] then $R$ must be serial (for every state $a \in W$ there must be some $b \in W$ such that $a R b$);
	\item [$T$,] then $R$ must be reflexive (for all $a \in W$, $a  R a $);
	\item [$4$,] then $R$ must be transitive (if $a R b R c$, then $a R c$);
	\item [$5$,] then $R$ must be euclidean (if $a R b$ and $a R c$, then $b R c$).
\end{description}

A model $(W,R,V)$ is a model for a logic $l$ if and only if $(W,R) \in F(l)$.
We call a formula satisfiable for a modal logic $l$, if it is satisfied in a state of a model for $l$. We call a formula valid for a modal logic $l$, if it is satisfied in all states of all models for $l$.

\begin{theorem}[Completeness, Finite Frame Property]\label{thm:completeness}
	A formula $\varphi$ is valid for  $l$ if and only if it is provable in $l$; $\varphi$ is satisfiable for $l$ if and only if it is satisfied in a finite model for $l$.
\end{theorem}

For the remainder of this paper we only consider finite Kripke models and frames.
For a finite model $\M = (W,R,V)$, we define $|\M| = |W| + |R|$.

\begin{definition}\label{def;complete}
	A formula $\varphi$ is called \emph{complete} for logic $l$ when for every $\psi \in L(P(\varphi))$, $\vdash_l \varphi \rightarrow \psi$ or $\vdash_l \varphi \rightarrow \neg \psi$.
	Formula $\varphi$ is \emph{incomplete} for $l$ if it is not complete for $l$.
\end{definition}

By Theorem \ref{thm:completeness}, $\varphi$ is complete for $l$ exactly when for every  $\psi \in L(P(\varphi))$, either $\psi$ or its negation is true at every (finite) pointed model for $l$ that satisfies $\varphi$.

\subsection{Bisimulation}

An important notion in \ML\ (and other areas) is that of bisimulation. Let $P$ be a (finite) set of propositional variables.
For Kripke models $\M = (W,R,V)$ and $\M' = (W',R',V')$, 
a non-empty relation $\mathcal{R}
\subseteq W\times W'$ is a \emph{bisimulation} (respectively, bisimulation modulo $P$) from $\M$ to $\M'$ when the following conditions are satisfied for all $(s,s')\in \mathcal{R}$:

\begin{itemize} 
	
	\item  $V(s)=V'(s')$ (resp. $V_P(s)=V_P'(s')$).
	
	\item For all $t\in W$ such that 
	$s R t$, there exists $t'\in W'$ such that $(t,t')\in \mathcal{R}$ and $s' R't'$.
	
	\item For all $t'\in W'$ such that 
	$s' R 't'$, there exists $t\in W$ such that $(t,t')\in \mathcal{R}$ and $s R t$.

\end{itemize}


We call pointed models $(\M,a), (\M',a')$ \emph{bisimilar} (resp. bisimilar modulo $P$) and write $(\M,a) \sim (\M',a')$ (resp. $(\M,a) \sim_P (\M',a')$) if there is a bisimulation (resp. bisimulation modulo $P$) $\R$ from $\M$ to $\M'$, such that $a \R a'$.
If $(\M,a)$ is a pointed model,
 and $P$ a set of propositional variables, 
then $Th_P(\M,a) = \{\varphi \in L(P) \mid \M,a \models \varphi \}$. We 
say that two pointed models are equivalent and write $(\M,a) \equiv_P (\M',a')$ 
when $Th_P(\M,a) = Th_P(\M',a')$.

The following simplification of the Hennessy-Milner Theorem \cite{hennessy1985algebraic} gives a very useful characterization of pointed model equivalence; Proposition \ref{prp:completeness_is_bisimilarity} is its direct consequence.

\begin{theorem}[Hennessy-Milner Theorem]\label{thm:eqbisim}
	If $(\M,a)$, $(\M',a')$ are finite pointed models, then \[(\M,a) \equiv_P (\M',a') \text{ if and only if } (\M,a) \sim_P (\M',a').\]
\end{theorem}

\begin{proposition}\label{prp:completeness_is_bisimilarity}
	A formula $\varphi$ is complete for a logic $l$ if and only if for every two pointed models $(\M,a)$ and $(\M',a')$ for $l$, if $\M,a \models \varphi$ and $\M',a'\models \varphi$, then $(\M,a) \sim_P (\M',a')$.
\end{proposition}

Paige and Tarjan in \cite{paige1987three} give an efficient algorithm for checking whether two pointed models are bisimilar. Theorem \ref{thm:bisimulation_is_easy} is a variation on their result to account for receiving the set $P$ of propositional variables as part of the algorithm's input.

\begin{theorem}\label{thm:bisimulation_is_easy}
	There is an algorithm which, given two pointed models $(\M,a)$ and $(\M',a')$ and a finite set of propositional variables $P$, determines whether $(\M,a) \sim_P (\M',a')$ in time $O(|P| \cdot (|\M| + |\M'|)\cdot \log (|\M| + |\M'|) )$.
\end{theorem}

\subsection{The Complexity of Satisfiability}

For logic $l$, the satisfiability problem for $l$, or $l$-satisfiability is the problem that asks, given a formula $\varphi$, if $\varphi$ is satisfiable. Similarly, the provability problem for $l$ asks if $\vdash_l \varphi$. 

The classical complexity results for \ML\ are due to Ladner \cite{ladnermodcomp}, who established \PSPACE-completeness for the satisfiability of \k, \t, \d, \kf, \df, and \sr\ and \NP-completeness for the satisfiability of \sv.  Halpern and R\^{e}go later characterized the \NP--\PSPACE\ gap  by the presence or absence of Negative Introspection \cite{Halpern2007Characterizing}, resulting in Theorem \ref{thm:ladhalp}.

\begin{theorem}\label{thm:ladhalp}
	If $l\in \{\k, \t, \d, \kf, \df, \sr\}$, then 
	$l$-provability is \PSPACE-complete and $l+5$-provability is \coNP-complete.
\end{theorem}

\section{The Completeness Problem and \NI}
\label{sec:NI} 

\emph{The completeness problem for $l$} asks, given a formula $\varphi$, if $\varphi$ is complete for $l$.
In this section, we explain how to adjust Halpern and R\^{e}go's techniques from \cite{Halpern2007Characterizing} to prove similar complexity bounds for the completeness problem for logics with \NI.
In the course of proving the \coNP\ upper bound for logics with \NI, Halpern and R\^{e}go give in \cite{Halpern2007Characterizing} a construction that provides a small model for a satisfiable formula. From parts of their construction, we can extract Lemma \ref{lem:NImodels} and Corollary \ref{cor:smallNImodels}. 

For a logic $l+5$, we call a pointed model $(\M,s)$ for $l+5$ \emph{flat} when 
\begin{itemize}
	\item $\M = (\{s\} \cup W,R,V)$;
	\item $R = R_1 \cup R_2$, where 
		$R_1 \subseteq \{s\} \times W$ and 
		$R_2$ is an equivalence relation on $W$; and
	\item if $l \in \{\t,\sr\}$, then $s \in W$.
\end{itemize}
Lemma \ref{lem:NImodels} informs us that flat models are a normal form for models of logics with axiom $5$.
%
\begin{lemma}\label{lem:NImodels}
	Every pointed $l+5$-model 
	$(\M,s)$ is bisimilar to a flat pointed $l+5$-model.
\end{lemma}

\begin{proof}
Let $W'$ be the set of states of $\M$ reachable from $s$ and $R$ the restriction of the accessibility relation of $\M$ on $W'$. It is easy to see that the identity relation is a bisimulation from $\M$ to $\M'$, so $(\M,s) \sim (\M',s)$; let 
$W = \{w \in W' \mid \exists w' R w \}$.
Therefore $W' = W \cup \{s\}$ and if $l \in \{\t,\sr\}$, then $s \in W$. 
Since $\M$ is an $l+5$-model, $R$ is euclidean. Therefore, the restriction of $R$ on $W$ is reflexive. This in turn means that $R$ is symmetric in $W$: if $a,b \in W$ and $a R b$, since $a R a$, we also have $b R a$. Finally, $R$ is transitive in $W$: if $a R b R c$ and $a,b,c \in W$, then $b R a$, so $a R c$. Therefore $R$ is an equivalence relation when restricted on $W$.
\end{proof}

The  construction from \cite{ladnermodcomp} and \cite{Halpern2007Characterizing} continues to filter the states of the flat model, resulting in a small model for a formula $\varphi$. 
Using this construction, Halpern and Rêgo prove Corollary \ref{cor:smallNImodels} \cite{Halpern2007Characterizing}; the \NP\ upper bound for $l+5$-satisfiability of Theorem \ref{thm:ladhalp} is a direct consequence.

\begin{corollary}
	\label{cor:smallNImodels}
	Formula
	$\varphi$ is $l+5$-satisfiable 
	if and only if 
	it is satisfied in a flat $l+5$-model of 
	$O(|\varphi|)$ states. 
\end{corollary}
Since we are asking whether a formula is complete, instead of whether it is satisfiable, we want to be able to find two small non-bisimilar models for $\varphi$ when $\varphi$ is incomplete.
For this, we need a characterization of bisimilarity between flat models.

\begin{lemma}\label{lem:flat_bisimulation}
	Flat pointed models
	$(\M,a) = (\{a\}\cup W,R,V)$ and $(\M',a') = (\{a'\}\cup W',R',V')$ are bisimilar modulo $P$ if and only if $V_P(a) = V_P(a')$ and: 
	\begin{itemize}
		\item for every $b \in W$, there is some $b' \in W'$ such that $V_P(b) = V'_P(b')$; 
		\item for every $b' \in W'$, there is some $b \in W'$ such that $V_P(b) = V'_P(b')$;
		\item 
		for every $b \in W$, if $a R b$, then there is a $b' \in W'$ such that  $a' R b'$ and $V_P(b) = V'_P(b')$; and
		\item  for every $b' \in W'$, if  $a' R b'$, then there is a $b \in W'$ such that $a R b$ and $V_P(b) = V'_P(b')$.
	\end{itemize}
\end{lemma}

\begin{proof}
If these conditions are met, we can define bisimulation $\R$ such that $a \R a'$ and for $b \in W$ and $b' \in W'$, $b \R b'$ iff $V_P(b) = V'_P(b')$; on the other hand, if there is a bisimulation, then it is not hard to see by the definition of bisimulation that these conditions hold --- for both claims, notice that the conditions above, given the form of the models, correspond exactly to the conditions from the definition of bisimulation.	
\end{proof}

This gives us Corollary \ref{cor:twosmallnonbis}, which is a useful characterization of incomplete formulas.

\begin{corollary}\label{cor:twosmallnonbis}
	Formula $\varphi$ is incomplete for $l+5$ if and only if it has two non-bisimilar flat pointed models for $l+5$ of at most $O(|\varphi|)$ states.
\end{corollary}
\begin{proof}
If $\varphi$ has two non-bisimilar pointed models for $l+5$, then by Theorem \ref{thm:eqbisim}, it is incomplete. On the other hand, if $\varphi$ is incomplete, again by Theorem \ref{thm:eqbisim} and  Lemma \ref{lem:NImodels}, $\varphi$ has two non-bisimilar flat pointed models, $(\M,a) = (\{a\}\cup W,R,V)$ and $(\M',a') = (\{a'\}\cup W',R',V')$. 
By Lemma \ref{lem:flat_bisimulation} and without loss of generality, we can distinguish three cases:
\begin{itemize}
	\item there is some $ p \in V_P(a) \setminus V_P(a')$: in this case, let $\psi = p$;
	\item 
	there is some $b \in W$, such that $a R b$ and for all $b' \in W'$ such that  $a' R b'$, $V_P(b) \neq V'_P(b')$: in this case, let 
	$$\psi = \Diamond (\bigwedge V_P(b) \wedge \neg \bigvee (P \setminus V_P(b)));$$
	\item there is some $b \in W$, such that for all $b' \in W'$, $V_P(b) \neq V'_P(b')$: in this case, let 
	$$\psi = \Diamond \Diamond (\bigwedge V_P(b) \wedge \neg \bigvee (P \setminus V_P(b))).$$
\end{itemize}
In all these cases, both $\varphi \wedge \psi$ and $\varphi \wedge \neg \psi$ are satisfiable and of size $O(|\varphi|)$, so by Corollary \ref{cor:smallNImodels}, each is satisfied in a non-bisimilar flat pointed model for $l+5$ of at most $O(|\varphi|)$ states.
 \end{proof}

Our first complexity result is a direct consequence of Corollary \ref{cor:twosmallnonbis}
and Theorem \ref{thm:bisimulation_is_easy}:

\begin{corollary}\label{cor:completeness_in_NP}
	The completeness problem for logic $l+5$ is in \coNP.
\end{corollary}


In the following, when $P$ is evident from the context, we will often omit any reference to it and instead of bisimulation modulo $P$, we will call the relation simply bisimulation.

\section{The Completeness Problem and Triviality}
\label{sec:completeness}

The first question we need to answer concerning the completeness problem for $l$ is whether there are any satisfiable and complete formulas for $l$.
If the answer is negative, then the problem is trivial. We examine this question with parameters the logic $l$ and whether $P$, the set of propositional variables we use, is empty or not. If for some logic $l$ the problem is nontrivial, then we give a complete formula $\varphi_P^l$ that 
uses exactly the propositional variables in $P$. We  see that for $P = \emptyset$, completeness can become trivial for another reason: for some logics, when $P = \emptyset$, all formulas are complete.
On the other hand, when $P \neq \emptyset$,  $\bigwedge P$ is incomplete for every logic.

\subsection{Completeness and \k}
Whether $P = \emptyset$ or not, completeness is nontrivial for \k\ and \kf:
let
$\varphi^\k_P = \varphi^\kf_P = \bigwedge P \wedge \Box \bot$ for every finite $P$.
Formula $\top$ is incomplete for \k\ and \kf.
\begin{lemma}
	Formula $\bigwedge P \wedge \Box \bot$ is complete and satisfiable for \k\ and for \kf.
\end{lemma}
\begin{proof}
	A model that satisfies $\varphi_P^\k$ is $\M = (\{a\},\emptyset,V)$, where $V(a) = P$. If there is another model $\M',a' \models \varphi^\k_P$, then $\M',a' \models \Box\bot$, so there are no accessible worlds from $a'$ in $\M'$;
	therefore, $\R = \{(a,a')\}$ is a bisimulation.
\end{proof}

Notice that if $\varphi$ is complete for $l$, then it is complete for every extension of $l$. Therefore,  $\varphi^\k_P$ is complete for all other logics.
However, we are looking for \emph{satisfiable and complete} formulas for the different logics, so finding one complete formula for \k\ is not enough. On the other hand, if $l'$ is an extension of $l$ (by a set of axioms) and a formula $\varphi$ is complete for $l$ and satisfiable for $l'$, then we know that $\varphi$ is satisfiable and complete for all logics between (and including) $l$ and $l'$.
Unfortunately, the following lemma demonstrates that we cannot use this convenient observation to reuse $\varphi^\k_P$ --- except perhaps for $\kv$ and $\kfv$, but these can be handled just as easily together with the remaining logics with Negative Introspection.

\subsection{Completeness and Consistency}

When 
$l$ has axiom $T$ or $D$, but not $4$ or $5$, 
$P$ determines if
a satisfiable formula is complete:
\begin{lemma}\label{lem:d_t_completeness}
	Let $l$ be either \d\ or \t.
	A satisfiable formula $\varphi \in L$ is complete with respect to $l$ if and only if $P(\varphi) = \emptyset$.
\end{lemma}
\begin{proof}
	When $P = \emptyset$, all models are bisimilar through the total bisimulation; therefore, all formulas $\varphi$, where $P(\varphi) = \emptyset$ are trivially complete. We now consider the case for $P \neq \emptyset$; notice that we can assume that $l = \d$, as $\d$ is contained in \t.
	Let the modal depth
	of $\varphi$ be $d$ and let $\M,a \models \varphi$, where $\M = (W,R,V)$; let $x \notin W^*$,
	 $a_0 = a$, and 
	$$
	\Pi_d = \{a_0\cdots a_k \in W^* \mid k \leq d \text{ and for all } 0 \leq i < k,\ a_i R a_{i+1}\}.
	$$ 
	Then, we define $\M'_1 = (W',R',V'_1)$ and $\M'_2 = (W',R',V'_2)$, where
	\begin{align*}
	W' \ = & \
	\Pi_d
	\cup \{x\};\\
	R' \ =& \ \{ (\alpha,\alpha b) \in W'^2 \mid b \in W \}\ \cup\ \{ (a_0 a_1\cdots a_d,x)\in W'^2 \}  \
	\cup \ \{(x,x) \}
	\\
	V'_i(\alpha b) \ =& \ V(b), \text{ for } i=1,2, \ 0 \leq |\alpha|< d ;
	\\
	V'_1(x) \ = & \ \emptyset; \ \text{ and } \ V'_2(x)\ =\ P .
	\end{align*}
	
	To prove that
	$\M'_1,a \models \varphi$ and $\M'_2,a \models \varphi$, we prove that for $\psi \in sub(\varphi)$, for every $i = 1,2$ and $w = a_0\cdots a_k \in \Pi_d$, where $k \leq d - md(\psi)$,
	$\M'_i,w \models \psi$ if and only if  $\M,a_k \models \psi$. We use induction on $\psi$. If $\psi$ is a literal or a constant, the claim is immediate and so are the cases of the $\wedge, \vee$ connectives. If $\psi = \Box \psi'$, then $md(\psi') = md(\psi)-1$; $\M'_i,w \models \psi$ iff for every $w R' w'$, $\M'_i,w' \models \psi'$ iff for every $a_k R' b$, $\M,b \models \psi'$ (by the Inductive Hypothesis) iff $\M,a_k \models \psi$; the case of $\psi = \Diamond  \psi'$ is symmetric.

	If $(\M'_1,a) \sim (\M'_2,a)$ through bisimulation $\R$ from $\M'_1$ to $\M'_2$, then 
	notice that in both models any sufficiently long path from $a$ will end up at $x$; therefore,
	by the conditions of bisimulation, 
	$x \R x$, which is
	a contradiction, since $V'_1(x)\neq V'_2(x)$.
	So, $\varphi$ is satisfied in two non-bisimilar models for $\d$.
\end{proof}

\subsection{Completeness, Consistency, and Positive Introspection}
For every finite $P$, let $\varphi_P^\df = \varphi_P^\sr  = \bigwedge P \wedge \Box \bigwedge P$. As the following lemma demonstrates, $\varphi_P^\df$ is a complete formula for \df\ and \sr.  
\begin{lemma}
	For every finite $P$, $\varphi_P^\df$ is  complete for \df\ and \sr;  all formulas in $L(\emptyset)$ are complete for \df\ and \sr.
\end{lemma}
\begin{proof}
	Let $\M,a \models \varphi_P^\df$ and $\M',a' \models \varphi_P^\df$; let $\R$ be the relation that connects all states of $\M$ that are reachable from $a$ (including $a$) to all states of $\M'$ that are reachable from $a'$ (including $a'$); it is not hard to verify that $\R$ is a bisimulation.
	Notice that if $P = \emptyset$, then $\varphi_P^\df$ is a tautology, thus all formulas are complete.
\end{proof}

It is  straightforward to see that $\varphi_P^\df$ is satisfiable for every modal logic $l$:  consider a model based on any frame for $l$, where $\bigwedge P$ holds at every state of the model. Therefore:

\begin{corollary}
	$\varphi^\df$ is satisfiable and complete for every extension of $\df$.\footnote{Although for the purposes of this paper we only consider a specific set of modal logics, it is interesting to note that the corollary 
		can be extended to a much larger class of logics.
	}
\end{corollary}

\subsection{Consistency and \NI}
For logic $l = l'+5$, let $\varphi^l_P = \bigwedge P \wedge \Diamond \Box \bigwedge P $. 
\begin{lemma}
	For any logic $l = l'+5$,
	$\varphi^l_P$ is a satisfiable complete formula for $l$.
\end{lemma}
\begin{proof}
	By Lemma \ref{lem:NImodels}, $\varphi^l_P$ is complete.    
	It is satisfied in $(\{a\},\{(a,a)\},V)$, where $V(a)=P$.
\end{proof}
When $P = \emptyset$, we can distinguish two cases. If $l' \in  \{\d,\df,\t,\sr\}$, then $\varphi^l_\emptyset$ is a tautology, therefore all formulas in $L(P)$ are complete for $l$.\footnote{This is also a corollary of Lemma \ref{lem:d_t_completeness}, as these are extensions of \d\ and \t.} If $l' \in \{\k,\kf\}$, by Lemma \ref{lem:NImodels}, an $l$-model would either satisfy $\varphi^l_P$ or $\Box \bot$, depending on whether the accessibility relation is empty or not.
 Therefore, if $P = \emptyset$ the completeness problem for \kv\ and \kfv\ is not trivial, but it is easy to solve: a formula with no propositional variables is complete for $l\in \{\kv, \kfv\}$ if it is satisfied in at most one of the two non-bisimilar modulo $\emptyset$ models for $l$.

\begin{corollary}
	If $P=\emptyset$, the completeness problem for \kv\ and \kfv\ is in \P.
\end{corollary}

\subsection{Completeness and Modal Logics}
A logic $l$ has a nontrivial completeness problem if for $P \neq \emptyset$, there are complete formulas for $l$. From the logics we examined, only \d\ and \t\  have trivial completeness problems.
Table \ref{table:tableoftriviality} summarizes the results of this section and of Section \ref{sec:complexity} regarding the completeness problem. As the table demonstrates, we can distinguish the following cases. For \k, the completeness problem is non-trivial and \PSPACE-complete; this does not change when we add axiom $4$. Once we add axiom $D$ to \k, but not $4$ or $5$, the completeness problem becomes trivial; adding the stronger axiom $T$ does not change the situation. Adding both $4$ and $D$ or $T$ to \k\ makes completeness \PSPACE-complete again, except when $P = \emptyset$. Regardless of other axioms, if the logic has \NI, completeness is \coNP-complete --- unless $P=\emptyset$, when the situation depends on whether the logic has $D$ (or the stronger $T$) or not.

\section{The Complexity of Completeness}
\label{sec:complexity}

\begin{table}[t]
	\centering
	\begin{tabular}{|r | c | c |}
		\hline 
		\ML\  \  & $P=\emptyset$ & $P \neq \emptyset$\\
		\hline 
		\k, \kf\  &  \PSPACE-complete		& 	\PSPACE-complete 	\\
		\d, \t\  & trivial (all)		&	trivial (none)	\\
		\df, \sr\  & 	trivial (all)	&	\PSPACE-complete	\\
		\kv, \kfv\  &	in \P &	\coNP-complete	\\
		$l+5$, $l \neq \k,\kf$  &	trivial (all) &	\coNP-complete	\\
		\hline
	\end{tabular}
	\caption{The complexity of the completeness problem for different modal logics. Trivial (all) indicates that all formulas in this case are complete for the logic; trivial (none) indicates that there is no satisfiable and complete formula for the logic.}
	\label{table:tableoftriviality}
\end{table}

Our main result is that for a modal logic $l$, the completeness problem has the same complexity as provability for $l$, as long as we allow for propositional variables in a formula and the completeness problem for $l$ is nontrivial (see also Table \ref{table:tableoftriviality}).
For the lower bounds, we consider hardness under polynomial-time reductions. 
As the hardness results are relative to complexity classes that include \coNP, these reductions suffice.

\subsection{A Lower Bound}

We  present a lower bound for the complexity of the completeness problem: we show that the completeness problem is at least as hard as provability for a logic, as long as it is nontrivial. 

\begin{theorem}\label{thm:hardness}
	Let $l$ be a logic that has a nontrivial completeness problem and let $C$ be a complexity class. If $l$-provability is $C$-hard, then the completeness problem for $l$ is $C$-hard.
\end{theorem}
\begin{proof}
	To prove the theorem we present a reduction from $l$-provability to the completeness problem for $l$. From a formula $\varphi$, the reduction constructs in polynomial time a formula $\varphi_c$, such that $\varphi$ is provable if and only is $\varphi_c$ is complete.
	For each logic $l$ with nontrivial  completeness and finite set of propositional variables $P$, in Section \ref{sec:completeness} we provided a complete formula $\varphi^l_P$.  This formula is satisfied in a model of at most two states, which can be generated in time $O(|P|)$. Let $(\M_l,a_l)$ be such a pointed model for $\varphi^l_P$. We assume that $P \neq \emptyset$.
	
	Any pointed model that satisfies $\varphi^l_P$ is bisimilar to $(\M_l,a_l)$. Given a formula $\varphi \in L(P)$, we can determine in linear time if $\M_l,a_l \models \varphi$. Then, there are two cases: 
	\begin{description} 
		\item[$\M_l,a_l \not\models \varphi$,] 
			in which case $\varphi$ is not provable and we set $\varphi_c = \bigwedge P$. 
		\item[$\M_l,a_l \models \varphi$,] so $\neg \varphi \wedge \varphi^l_P$ is not satisfiable, in which case we set $\varphi_c = \varphi \rightarrow \varphi^l_P$.
			We demonstrate that $\varphi$ is provable if and only if $\varphi \rightarrow \varphi^l_P$ is complete. 
			\subitem
			If $\varphi$ is provable, then $\varphi \rightarrow \varphi^l_P$ is equivalent to $\varphi^l_P$, which is complete. 
			\subitem 
			On the other hand, if $\varphi \rightarrow \varphi^l_P$ is complete and $(\M,a)$ is any pointed model, we show that $\M,a  \models \varphi$, implying that if $\varphi \rightarrow \varphi^l_P$ is complete, then $\varphi$ is provable. If $(M,a) \sim_P (M_l,a_l)$, then from our assumptions $\M,a \not \models \neg \varphi$, thus $\M,a  \models  \varphi$. On the other hand, if $(M,a) \not\sim_P (M_l,a_l)$, since $(M_l,a_l) \models \varphi \rightarrow \varphi^l_P$ and $\varphi \rightarrow \varphi^l_P$ is complete, $\M,a \not \models \varphi \rightarrow \varphi^l_P$, therefore $\M,a \models \varphi$. \qedhere
	\end{description}
\end{proof}


Theorem \ref{thm:hardness} applies to more than the modal logics that we have defined in Section \ref{sec:back}.
For Propositional Logic, completeness amounts to the problem of determining whether a formula does not have \emph{two} distinct satisfying assignments, therefore it is \coNP-complete. 
By similar reasoning, completeness for First-order Logic is undecidable, as satisfiability is undecidable.

\subsection{Upper Bounds}


The easiest cases are the logics with axiom $5$. 
Immediately from
Theorem \ref{thm:hardness} and Corollary \ref{cor:completeness_in_NP}:

\begin{proposition}\label{prp:completenessforNI}
	The completeness problem for logic $l+5$ is \co\NP-complete.
\end{proposition}

For the logics without axiom $5$, by Theorem \ref{thm:ladhalp}, satisfiability and provability are both \PSPACE-complete. So, completeness is \PSPACE-hard, if it is nontrivial. It remains to show that it is also in \PSPACE.
To this end we present a procedure that decides completeness for a modal formula. We call it the CC Procedure. Parts of this procedure are similar to the tableaux by Fitting \cite{melfittableaux} and Massacci \cite{Massacci2000a} for \ML, in that the procedure explores local views of a tableau branch. For more on  tableaux the reader can see \cite{DAgostino1999}. 
The CC Procedure is a non-deterministic polynomial time algorithm that uses an oracle from \PSPACE. It accepts exactly the incomplete formulas, thus establishing that the completeness problem for these logics is in \PSPACE.
We have treated the case for logics with axiom $5$, and the completeness problem for $\d$ and \t\ is trivial. Therefore, form now on, we fix a logic $l$ that can either be $\k$, or have axiom $4$ and be one of \kf, \df, and \sr. 



\subsubsection*{The CC Procedure for \ML\ $l$ on $\varphi$}
Intuitively, the procedure tries to construct two models for $\varphi$ and at the same time demonstrate that these models are not bisimilar. 
We first give a few definitions that we need to describe the procedure.

For our procedure, \emph{states} are sets of formulas from $\overline{sub}(\varphi)$.
The procedure generates structures that we call \emph{views}.
A view $S$ is a pair $\left(p(S),C(S)\right)$ of a (possibly empty) set $C(S)$ of states, that are called the \emph{children-states} of $S$ and a distinguished state $p(S)$ called the \emph{parent-state} of $S$.
Each view is allowed to have up to $|\varphi|$ children-states.


\begin{definition}
We call a set $s$ of formulas   \emph{$l$-closed} if the following conditions hold:
\begin{itemize}
	\item 
	if $\varphi_1 \wedge \varphi_2 \in  s$, then $\varphi_1, \varphi_2 \in  s$; 
	\item 
	if $\varphi_1 \vee \varphi_2 \in  s$, then $\varphi_1 \in s$ or $\varphi_2 \in  s$; 
	\item 
	if $\Box \psi \in s$ and $l$ has axiom $T$, then $\psi \in s$; 
	\item 
	for every $p \in P$, either $p \in s$ or $\neg p \in s$.
\end{itemize}

We call a view $S$ \emph{$l$-complete} (or \emph{complete} if $l$ is fixed) if the following conditions hold:
\begin{itemize}
	\item the parent-state and every child-state of that view are $l$-closed;
	\item for every $\Diamond \psi \in p(S)$, $\psi \in \bigcup C(S)$;
	\item for every $\Box \psi \in p(S)$, $\psi \in \bigcap C(S)$;
	\item if $l$ has axiom $4$, then for every $\Box \psi \in p(S)$, $\Box\psi \in \bigcap C(S)$;
	\item if $l$ has axiom $D$, then $C(S) \neq \emptyset$.
\end{itemize}

For state $a$,  $th(a) = \bigwedge a$. 
A state $a \subseteq \overline{sub}(\varphi)$ is 
\emph{maximal}
if it is a maximally consistent subset of 
$\overline{sub}(\varphi)$.
A child-state $c$ of a view $S$ is \k-maximal when it is a maximally consistent subset of 
$\overline{sub}_{d}(\varphi)$, where $d = \max \{ md(c') \mid c' \in C(S) \}$.
A view $S$ is \emph{consistent} when every state of $S$ is a consistent set of formulas.
A view $S'$ \emph{completes} view $S$ when: 
	$S'$ is $l$-complete; 
	$p(S) \subseteq p(S')$; 
	for every $a \in C(S)$ there is an $a' \in C(S')$ such that $a \subseteq a'$; 
	and: if $l = \k$, then every $a' \in C(S')$ is $\k$-maximal; if $l$  has axiom $4$, then every $a' \in C(S')$ is maximal.
\end{definition}

A view gives a local view of a model, as long as it is consistent. 
The procedure generates  views and ensures
that they are complete --- so that all relevant information is present in each view --- and consistent --- so that the view  indeed represents parts of a model.\footnote{As the reader may notice, the procedure does not have to generate a full view. Only generating a maximal state $c$ and verifying that it can be a child-state in a view that has $a$ as parent-state suffices.}
If the parent-state can represent two non-bisimilar states of two models (say, $s$ and $t$), then the procedure should be able to provide a child, representing a state accessible from $s$ or $t$ that is not bisimilar to any state accessible from $s$ or $t$, respectively. 
Since the states are (\k-)maximal, two states that are not identical can only be satisfied in non-bisimilar models.
The procedure is given in Table \ref{tab:CCprocedure}.

\begin{table}
	\centering
	\begin{tabular}{|r|p{.7\textwidth}|}
		\hline 
		Initial Conditions:
		&
		Non-deterministically generate maximal states $a$ and $b$ that include $\varphi$; if there are none,  then return ``\texttt{reject}''.\\
		& If $a \neq b$, then return ``\texttt{accept}.''\\
		& Initialize $N$ to $\constant$.
		\\
		%
		Construction:
		&
		Non-deterministically generate a consistent view $S$ that completes $(a,\emptyset)$, having up to $|\varphi|$ children-states. 
		\\
		Condition: &
		If $C(S) = \emptyset$, then return ``\texttt{reject}.''
		\\ &
		If there is a child-state $c \in C(S)$, such that $\not \vdash_l th(a) \rightarrow \Diamond th(c)$, then return ``\texttt{accept}.''
		\\
		Next: &
		Otherwise, non-deterministically pick a child $c \in C(S)$ and set $a:=c$.
		\\ & 
		If $N>0$, then set $N:=N-1$ and continue from ``Construction.''
		\\
		If $N = 0$, & then return ``\texttt{reject}''.\\
		\hline
	\end{tabular}
	\caption{The CC Procedure on $\varphi$ for logic $l \in \{\k, \kf, \df, \sr \}$.}
	\label{tab:CCprocedure}
\end{table}

This section's main theorem is Theorem \ref{thm:main} and informs us our procedure can determine the completeness of formula $\varphi$ in at most $\constant$ steps. That the completeness problem for logics without axiom 5 is in \PSPACE\ is a direct corollary.

\begin{theorem}\label{thm:main}
	The CC Procedure accepts $\varphi$ if and only if $\varphi$ is incomplete.
\end{theorem}

\begin{proof}[Sketch of Proof]
	We sketch the proof for the case where $l = \k$. The full proof is given in Section \ref{sec:proof}.
	We first assume that $\varphi$ is incomplete.
	Therefore, there are two non-bisimilar pointed models $(\M_1,a_0)$ and $(\M_2,b_0)$ that satisfy $\varphi$ --- w.l.o.g. we assume that $\M_1 = \M_2$ with accessibility relation $R$, and we omit them from the description.
	The procedure can ensure that at every step $i$ in the resulting play, there are non-bisimilar states $a_i,b_i$, that satisfy the chosen child-state --- this is the invariance condition.
%
	This condition is already true at the beginning and it can be maintained recursively by the procedure on each step by choosing a state $a_i$ accessible from $a_{i-1}$ (or, similarly, state $b_i$ accessible from $b_{i-1}$) that is non-bisimilar to all states accessible from $b_{i-1}$ --- $a_i$ always exists due to the invariance. Then, the procedure can produce and choose a maximal child-state $a$ that is satisfied in $a_i$. Then, either $\Diamond th(a)$ is not derivable from the parent-state, or it is satisfied in some $b_i$ that is non-bisimilar to $a_i$, and the invariance is maintained.
	Notice that the procedure can reduce the modal depth of the states every 
	time the ``Construction'' step takes place and when the modal depth of the parent-state is 0, it can just construct child-state $P$ that is not derivable from that parent-state.

	We now assume that $\varphi$ is complete. 
	Then, simply notice that if a view is $\k$-complete and consistent, then it represents a state in a model and its accessible states. If a parent-state is complete, then for every child $a$ of a consistent $\k$-complete view, $\Diamond th(a)$ can be derived from the parent-states. 
	For more details, see Section \ref{sec:proof}.
\end{proof}

\begin{corollary}
	The completeness problem for \k, \kf, \df, and \sr\ is \PSPACE-complete.
\end{corollary}
\begin{proof}
	\PSPACE-hardness is a consequence of Theorem \ref{thm:hardness}. 
	The CC Procedure is a non-deterministic polynomial-time algorithm with an oracle from \PSPACE. 
	Each condition that it needs to check is either a closure condition or a condition for the consistency or provability of formulas of polynomial size with respect to $|\varphi|$; therefore, they can be verified either directly or with an oracle from \PSPACE.
	Thus, the completeness problem for these logics is in 
	$ \coNP^\PSPACE = \PSPACE$.
\end{proof}

\section{The Proof of Theorem \ref{thm:main}}
\label{sec:proof}

We prove that the CC Procedure has a way to accept $\varphi$ if and only if $\varphi$ is satisfied in two non-bisimilar models. By Theorem \ref{thm:eqbisim}, the theorem follows.

\paragraph{We first assume that there are two non-bisimilar pointed models $(A,w)$ and $(B,w')$, such that $A,w \models \varphi$ and $B, w' \models \varphi$. We prove that the CC Process accepts $\varphi$ in $\constant$ steps.} 
We call these models the underlying models; the states of the underlying models are called model states to distinguish them from states that the process uses.
%
Let $A = (W^A,R^A,V^A)$ and $B = (W^B,R^B,V^B)$; we can assume that $W^A \cap W^B = \emptyset$.
Let $f: W^A \times W^B \to W^A \cup W^B$ be a partial function that maps every pair $(s,t)$ of non-bisimilar pairs to a model state $c$ accessible from $s$ or $t$ that is non-bisimilar to every state accessible from $t$ or $s$, respectively.
We call $f$ a choice-function.
We can see that the procedure can maintain that the maximal state it generates each time is satisfied in two non-bisimilar states $s,t$, one from $A$ and the other from $B$, respectively: at the beginning these are $w$ and $w'$. At every step, the procedure can pick a child $c$ that is satisfied in $f(s,t)$. If $\not \vdash_l th(a) \rightarrow \Diamond th(c)$, then the procedure terminates and accepts the input. Otherwise, $c$ is satisfied in $f(s,t)$ and in another state that is non-bisimilar to $f(s,t)$. Let that other state be called a counterpart of $f(s,t)$.

If $l = \k$, then at every step, the procedure can reduce the modal depth of $a$, and therefore, after at most $|\varphi|$ steps, the procedure can simply choose $P = P(\varphi)$ as a state. Since $\Diamond \bigwedge P$ is not derivable from any consistent set of modal depth $0$, the procedure can terminate and accept the input. We now assume that $l \neq \k$.

We demonstrate that if $\varphi$ is incomplete, then the CC Procedure will accept $\varphi$ after a finite number of steps. 
As we have seen above, the procedure, given non-bisimilar pointed models $(A,a)$ and $(B,b)$ of $\varphi$, always has a child to play according to $f$.
For convenience, we can assume that models $A$ and $B$ have no cycles, so the choice-function never repeats  a choice during a process run.
If for every choice of $f$, the process does not terminate, then
we show that $(A,w) \sim (B,w')$, reaching a contradiction. Let $\R = \sim \cup Z$, where $\sim$ is the bisimilarity relation between the states of $A$ and the states of $B$, and $x Z y$ when for some choice-function, there is an infinite execution of the procedure, in which $y$ is a counterpart of $x$, or $x$ a counterpart of $y$.
If $x \R y$, either 
	$(A,x) \sim (B,y)$, so $V^A_P(x) = V^B_P(y)$, or 
	$x Z y$, so, again, $V^A_P(x) = V^B_P(y)$, since $x$ and $y$ satisfy he same maximal state.
If $x \R y$ and $x R^A x'$, then if $(A,x) \sim (B,y)$, immediately there is some $y R^B y'$ so that $(A,x') \sim (B,y')$; 
if $x$ is a counterpart of $y$ or $y$ is a counterpart of $x$ during a non-terminating run, then
for every 
$x'$ accessible from $x$ (the case is symmetric for a $y'$ accessible from $y$), either $x'$ is bisimilar to some $y'$ accessible from $y$, or 
we can alter the choice-function $f$ that the procedure uses so that
$x' = f(x,y)$. 
Since for that altered $f$, the procedure does not terminate, $x'$ has a counterpart as well. 
Therefore, the bisimulation conditions are satisfied and $\R$ is a bisimulation.
If for every choice-function, the procedure never terminates, then $(A,w) \sim (B,w')$, and we have reached a contradiction.
Therefore, there is 
	a choice-function $f$ that ensures the procedure terminates after a finite number of steps. We call that number of steps the length of choice-function $f$.

For 
every state $a$,
let 
$$D(a) = \{  \Diamond \psi \in 
a 
\} $$ 
and 
$$B(a) = \{ \Box \psi \in  a  \}.$$
Then, 
$0 \geq |D(a)| \leq k_1$ and $0 \leq |B(a)| \leq k_2$, where $0 \leq k_1 + k_2 \leq |\varphi| - 1$. 
Notice that according to the definition of $f$ above, as the 
process runs,
$D(a)$ decreases and $B(a)$ increases --- though, not necessarily strictly. 

\begin{lemma}\label{lem:three_max_states}
	Let $l \in \{\kf, \df, \sr \}$ and let $a,b,c$ be maximal states. If $B(a) = B(b)$, $D(a) = D(b)$, $\vdash th(a) \rightarrow_l \Diamond th(c)$, and $\not \vdash_l th(b) \rightarrow \Diamond th(c)$, then $c = a \neq b$ and $l = \sr$.
\end{lemma}

\begin{proof}
	If $\not \vdash_l th(b) \rightarrow \Diamond th(c)$, then $b \cup \{\Box \neg th(c) \}$ is consistent, so satisfiable in a pointed model $(\M,s)$.
	Let $\M'$ be the result of adding a new state $s'$ to $\M$ with the same accessible states as $s$, but satisfying exactly the propositional variables in $a$. This is an $l$-model, unless $l = \sr$, and it is the case that $\M',s' \models \Box \neg th(c)$; furthermore, by straightforward induction on $\psi$, for every $\psi \in \overline{sub}(\varphi)$, $\M',s' \models \psi$ iff $\psi \in a$. 
	Therefore, $\M',s' \models th(a)$ and $\M',s' \not\models \Diamond th(c)$, so $l = \sr$.
	By making $s'$ accessible from itself, we get an \sr-model $\M''$.
	For every formula $\psi \in \overline{sub}(\varphi)$, $\M'',s' \models \psi$ iff $\psi \in a$: literals and boolean cases are immediate; if $\Diamond \psi$ or $\neg \Box \psi$ is in $a$, it is also in $b$, and $\M'',s \models th(b)$; if $\Box \psi \in a$, then $\psi \in a$ so by I.H. $\M'',s' \models \psi$, and $\M'',s \models \Box \psi$, so $\psi$ holds at every accessible state from $s'$ ---  and similarly for $\neg \Diamond \psi$. Therefore, $\M'',s' \models th(a)$, so $\M'',s' \models \Diamond th(c)$. Since the only state accessible from $s'$ and not from $s$ is $s'$, $\M'',s' \models th(c)$, and therefore, $a=c$.
\end{proof}

We can safely assume that the procedure never repeats the same choice of child-state --- otherwise, it could continue from the second repetition and shorten its run.
If during an execution, the CC Procedure picks states $a$, and in a following step, a state $b$, so that $B(a) = B(b)$ and $D(a) =  D(b)$, and immediately after $b$ the procedure picks child-state $c$, we claim that either the procedure could pick $c$ right after $a$ without affecting its run, or $a$ and $b$ are consecutive picked states and after picking $c$, the procedure terminates. 
Since $c$ can be a child-state for a view that has $b$ as parent-state, it satisfies all necessary closure conditions for $l$-complete views, so it can appear as a child-state for a view that has $a$ as parent-state. 
If $\not\vdash_l th(a) \rightarrow \Diamond th(c)$, then the procedure can pick $c$ right after $a$ and terminate immediately; 
if $\vdash_l th(a) \rightarrow \Diamond th(c)$, but $\not\vdash_l th(b) \rightarrow \Diamond th(c)$, then the procedure terminates at $c$ and, by Lemma \ref{lem:three_max_states}, $l = \sr$ and $a = c$.
If $a$ and $b$ are not consecutive states, then there is a  maximal state $a'$ picked after $a$ and before $b$, so that $B(a') = B(b)$ and $D(a') =  D(b)$.
Similarly to the above, $a' = c$, and therefore, $a = a'$ --- so, the procedure repeated the same child-state choice.
Therefore, a minimal-length choice function can ensure that the CC Procedure
terminates after  $|\varphi| + 2$ steps.

\paragraph{On the other hand, we prove that if $\varphi$ is complete, then the CC Procedure can never accept $\varphi$.} For this, we use the following two lemmata:

\begin{lemma}\label{lem:consistentchildren}
	If a view $S$ is consistent and complete
	and $C(S) \neq \emptyset$, 
	then
	\begin{itemize}
		\item if $l$ does not have axiom $4$ ($l = \k$), then the following formula is consistent: 
		\[th(p(S)) \wedge \bigwedge_{c\in C(S)} \Diamond th(c) \wedge \Box \bigvee_{c\in C(S)} th(c); \]
		\item if $l$ has axiom $4$ ($l \in \{\kf, \df, \sr\}$), then the following formula is consistent:
		\[th(p(S)) \wedge \bigwedge_{c\in C(S)} \Diamond th(c). \] 
	\end{itemize}
\end{lemma}

\begin{proof}
	For $c \in C(S)$, let $\M_c = (W_c,R_c,V_c)$ and $a_c \in W_c$, such that $\M_c,a_c \models th(c)$; we assume $p(S) \notin W_c$. Then let $\M = (W,R,V)$, where 
	\begin{align*}
	W \ =& \ \{p(S) \} \cup \bigcup_{c\in C(S)} W_c,
	&R' \ =& \ \{(p(S),a_c) \mid c \in C(S) \} \cup \bigcup_{c\in C(S)} R_c,
	\end{align*}
	$R$ is the transitive (if $l$ has  axiom $4$ and not $T$) or reflexive and transitive (if $l$ has both axioms) closure of $R'$, or just $R'$ (if $l$ has neither axiom) --- if $C(S) = \emptyset$, then the lemma is true immediately for $l \neq \k$, so we can assume that $R'$ is serial when $l = \d$ ---
	and
	\[ V(p(S)) = \{ p \in P \mid p \in p(S) \}\]
	and for $b \in W_c$, $ V(b) =  V_c (b).$
	Now,  $\M,a_c \models th(c)$, since it is not hard to see that 
	$(\M,a_c) \sim (\M_c,a_c)$.
	By straightforward induction on $\psi$, we can see that for all  $\psi \in p(S)$,
	$\M,p(S) \models \psi$, from which we can
	conclude that 
	\begin{align*}
	\M,p(S) \models& th(p(S)) \wedge \bigwedge_{c\in C} \Diamond th(c) \wedge \Box \bigvee_{c\in C} th(c),  &&\text{ if $l = \k$; or} \\
	\M,p(S) \models& th(p(S)) \wedge \bigwedge_{c\in C} \Diamond th(c),  &&\text{ if $l$ has axiom 4.}  \tag*{\qedhere} 
	\end{align*}
\end{proof}

\begin{lemma}\label{lem:consistentchildren4}
	Let
		$s$ be a consistent, and complete state, and for $l \neq \k$, also a maximal state; 
		$d$ a maximal state; and
	$\psi$ a formula.
If 
\begin{itemize} 
	\item $\vdash_l th(s) \rightarrow \Diamond th(d)$,
	\item 
	$th(d)$ is not equivalent to $th(s)$, and 
	\item 
	$d\cup \{\Box \psi\}$ is consistent,
	\end{itemize}
	then 
	$th(s) \wedge \Box (\neg th(d)\vee \Box\psi)$ is consistent.
\end{lemma}
\begin{proof}
	If $l = \k$, then we can use Lemma \ref{lem:consistentchildren}; otherwise,
	we can use a similar construction as for Lemma \ref{lem:consistentchildren}.
	Let $$C = \{ 
		s' \subseteq \overline{sub}(\varphi) \mid\ 
		s' \text{ is maximal and }
		\vdash_l th(s) \rightarrow \Diamond th(s') 
		\};$$
	Notice that, due to maximality, if for some $a \in C$, $th(a) \wedge th(d)$ are consistent, then $a = d$.	
	Let $\M = (W,R,V)$, where: $W = \{s\} \cup C$; for $a,b \in W$, $a R b$ iff for every $\Box \chi \in a$, $\chi, \Box \chi \in b$; and $V(a) = P\cap a$.
	$\M$ is an $l$-model: $R$ is transitive, since if $\Box \chi \in a$ and $a R b$, also $\Box \chi \in b$; if $l$ has axiom $D$, then $R$ is serial, because for $a \in W$, $\{ \chi, \Box \chi \mid \Box\chi \in a \}$ is consistent; and if $l$ has axiom $T$, then $R$ is reflexive, since for $a$ to be maximally consistent, if $\Box \chi \in a$, then $\chi \in a$.
	
	By induction on $\chi$, for every $\chi \in \overline{sub}(\varphi)$ and $a \in W$, $\chi \in a$ iff $\M,a \models \chi$: constants, literals and boolean connectives are immediate; if $\Box \chi \in a$, then for every $a R b$, $\chi \in b$, so $\M,a \models \Box \chi$; if $\Diamond \chi \in a$, then  $v = \{ \chi', \Box \chi' \mid \Box\chi' \in a \} \cup \{\chi\}$ is consistent and $\vdash_l th(a) \rightarrow \Diamond th(v)$, so $\vdash_l th(s) \rightarrow \Diamond th(v)$ (by axiom 4), therefore there is some $b \supseteq v$ in $C$ and $a R b$.
	
	Let $D,x \models th(d) \wedge \Box \psi$, where $D = (W_d,R_d,V_d)$ is an $l$-model;
	let $\M' = (W\cup W_d,R',V')$, where $V'(a) = V(a)$ if $a \in W$ and $V'(a) = V_d(a)$ otherwise, and $R'$ is the transitive closure of
	\[
	\{(a,b) \in R \mid a = d \text{ or } b \neq d \} \cup 
	\{ (a,x) \mid a \neq d \text{ and } a R d \} \cup R_d. 
	\] 
	Then, by induction on the formulas, for every $\chi \in \overline{sub}(\varphi)$ and $a \in W$, $\M',a \models \chi$ iff $\M,a \models \chi$: constants, literals, boolean connectives are immediate; if $\chi = \Box \chi'$, then if not $a R d$, nothing changed and if $a R d$, $\chi, \chi' \in d$, so 
	$D,x \models \chi,\chi'$, meaning that if $a R y$, then $\M',y \models \chi'$; the case for $\chi = \Diamond \chi'$ is more straightforward.  
	Therefore, $\M',s \models th(s) \wedge \Box(\neg th(d) \vee \Box \psi)$. 
\end{proof}

\begin{lemma}\label{lem:completetocomplete}
	For a consistent view $S$ that completes itself, for every child $c \in C(S)$, 
	if $th(p(S))$ is complete, then so is $th(c)$.
\end{lemma}

\begin{proof}
	If $th(p(S))$ is complete and $th(c)$ is not, then $\not \vdash_l th(c) \rightarrow th(p(S))$ and there is some $\psi$ such that $th(c) \not \vdash \psi $ and $th(c) \not \vdash \neg\psi $. 
	Since $c$ is (\k-)maximal, for every other $b \in C(S)$, $b \cup \{ \neg th(c) \}$ is consistent.
	Therefore 
	there is a
	a  consistent view $S_1$ that completes $S$ and has a child-state $c_1 \supseteq \{\psi\} \cup c $ and  
	a 
	consistent view $S_2$ that completes $S$, which has
	a child-state $c_2  \supseteq \{ \neg \psi\} \cup c $ and $\neg th(c)$ is an element of every other child-state of $S_2$. 
	By Lemma \ref{lem:consistentchildren}, 
	if $l = \k$, then 
	$th(p(S)) \wedge \Diamond (th(c) \wedge \psi)$ is consistent and so is $th(p(S)) \wedge \neg \Diamond (th(c) \wedge \psi)$, which is a contradiction.
	
	If $l$ has axiom 4, then 
	for every other child (or maximally consistent subset of $\overline{sub}(\varphi)$) $c'$, $\vdash_l th(c') \rightarrow \neg th(c)$;
	$c$ is maximal and 
	we can assume $\psi = \Box \chi$ --- due to maximality, $c$ determines the truth-values of literals, and if it determined the truth-value of all boxed formulas, it would determine the truth-value of all formulas and would be complete.
	By Lemma
	\ref{lem:consistentchildren}, $th(p(S)) \wedge \Diamond (th(c) \wedge \neg \psi)$ is consistent. 
	By Lemma \ref{lem:consistentchildren4}, so is 
	$th(p(S)) \wedge \Box (\neg th(c) \vee \psi)$, again a contradiction.
\end{proof}	

By Lemma \ref{lem:completetocomplete}, all parent-states that appear during a run are complete. If at some point, the process picks a child-state $c$ and $a$ is the parent-state, then by Lemma \ref{lem:consistentchildren}, $th(a) \wedge \Diamond th(c)$ is consistent; since $a$ is complete, $\vdash_l th(a) \rightarrow \Diamond th(c)$. Therefore, there is no way for the procedure to accept if the input formula is complete.
\qed

\section{Variations and Other Considerations}
\label{sec:conclusions} 

There are several variations one may consider for the completeness problem.
One may define the  completeness of a formula in a different way, or consider a different logic, depending on the intended application. 
One may also wonder whether we could attempt a solution to the completeness problem by using Fine's normal forms
\cite{fine1975normal}.


%

\subsection{Satisfiable and Complete Formulas}

It may be more appropriate, depending on the case, to check whether a formula is \emph{satisfiable and complete}. In this case, if the modal logic does not have axiom 5, we can simply alter the CC Process so that it accepts right away if the formula is not satisfiable. Therefore, the problem remains in \PSPACE; for \PSPACE-completeness, notice that the reduction for Theorem \ref{thm:hardness} constructs  satisfiable formulas.

For logics with  axiom $5$ (and plain Propositional Logic), the language 
of satisfiable and complete formulas is \US-complete, where a language $U$ is in \US\ 
when
there is a nondeterministic Turing machine $T$, so that for every instance $x$ of $U$, $x \in U$ if and only if $T$ has exactly one accepting computation path for $x$\footnote{We note that \US\ is different from \UP; for \UP, if $T$ has an accepting path for $x$, then it is \emph{guaranteed} that it has a unique accepting path for $x$.} \cite{Blass198280}: Unique\SAT\ is a complete problem for \US\ and a special case of this variation of the completeness problem.

\subsection{Completeness with Respect to a Model}

A natural variation of the completeness problem would be to consider completeness of a formula over a satisfying model. That is, the problem would ask: 
given a formula $\varphi$ and pointed model $(\M,s)$, such that $\M,s \models \varphi$, is the formula complete? 
For this variation, we are given one of $\varphi$'s pointed models, so it is a reasonable expectation that the problem became easier. Note that in many cases, this problem may even be more natural than the original one, as we are now testing whether the formula completely describes the pointed model (that is, whether the formula is characteristic for the model).

Unfortunately, 
this variation has exactly the same complexity as the original completeness problem. 
We can easily reduce completeness with respect to a model to plain completeness by dropping the model from the input.
On the other hand, the reduction from provability to completeness of Section \ref{sec:complexity} still works in this case, as it can easily be adjusted to additionally provide the satisfying model of the complete formula  $\varphi_P^l$.

\subsection{Completeness and Normal Forms for Modal Logic}

In \cite{fine1975normal}, Fine introduced normal forms for \ML. 
The sets $F_P^d$ are  defined recursively on the depth $d$, which is a nonnegative integer,  and depend on the set of propositional variables  $P$  (we use a variation on the presentation from \cite{moss2007finite}):
\begin{align*}
F_P^0 = & \left\{ \bigwedge_{p\in S} p \wedge \bigwedge_{p\notin S} \neg p \mid S \subseteq P \right\}; \ \text{ and}\\
F_P^{d+1} = & \left\{ \varphi_0 \wedge \bigwedge_{\varphi\in S} \Diamond \varphi \wedge \Box\bigvee_{\varphi\in S}  \varphi  \mid S \subseteq F_P^d, \ \varphi_0 \in F_P^0 \right\}.
\end{align*}
For example, formula $\varphi^\k_P$ from Section \ref{sec:completeness} is a normal form in $F_P^1$.
\begin{theorem}[from \cite{fine1975normal}]\label{thm:fine}
	For every  modal formula $\varphi$ of modal depth at most $d$, if $\varphi$ is consistent for \k, then there is some $S \subseteq F_P^d$, so that $\vdash_\k \varphi \leftrightarrow \bigvee S$.
\end{theorem}
Furthermore, as Fine \cite{fine1975normal} demonstrated, normal forms are mutually exclusive: no two distinct normal forms from $F_P^d$ can be true at the same state of a model.
Normal forms are not necessarily complete by our definition (for example, consider $p \wedge \Diamond p \wedge \Box p $ for $P = \{p\}$), but, at least  for \k, it is not hard to distinguish the complete ones; by induction on $d$, 
	$\varphi \in F_P^d$ is complete for \k\ if and only if $md(\varphi) < d$. 
Therefore, for \k, the satisfiable and complete formulas are exactly the ones that are equivalent to such a complete normal form.
However, we cannot use this observation to test formulas for completeness by guessing a complete normal form and verifying that it is equivalent to our input formula, 
as normal forms can be of very large size: 
$|F_P^0| = 2^{|P|}$; $|F_P^{d+1}| = |P| \cdot 2^{|F_P^{d}|}$; and if $\psi \in F_P^d$, $|\psi|$ can be up to $|P| + 2|F_P^{d-1}|$. 
We would be guaranteed a normal form of reasonable (that is, polynomial w.r.to $|\varphi|$) size to compare to $\varphi$ only if $\varphi$ uses a small (logarithmic with respect to $|\varphi|$) number of variables and its modal depth is very small compared to $|\varphi|$ (that is, $md(\varphi) = O(\log^*(|\varphi|))$).

\subsection{Completeness up to Depth}

Fine's normal forms \cite{fine1975normal} can inspire us to consider a relaxation of the definition of completeness. We call a formula $\varphi$ \emph{complete up to its depth} for a logic $l$ exactly when for every formula $\psi \in L(P(\varphi))$ of modal depth at most $md(\varphi)$, either $\vdash_l \varphi \rightarrow \psi$ or $\vdash_l \varphi \rightarrow \neg \psi$.
Immediately from Theorem \ref{thm:fine}:
\begin{lemma}\label{lem:norm_forms_are_complete_up_to_depth}
		All normal forms are complete up to their depths.
\end{lemma}

\begin{lemma}\label{lem:deep_completeness_is_character}
	Formula $\varphi$ is satisfiable and complete up to its depth for logic $l$ if and only if it is equivalent in $l$ to a normal form from $F_P^{md(\varphi)}$.
\end{lemma}
\begin{proof}
	From Theorem \ref{thm:fine}, if $\varphi$ is satisfiable, then it is equivalent to  some $\bigvee S$, where
$S \subseteq F_P^{md(\varphi)}$, but if it is also complete up to its depth, then it can derive a the normal form $\psi \in S$; so, $\vdash_l \varphi \rightarrow \psi$, but also $\vdash_l \psi \rightarrow \bigvee S$ and $\bigvee S$ is equivalent to $\varphi$. 
For the other direction, notice that every normal form in $F_P^{md(\varphi)}$ is either complete or has the same modal depth as $\varphi$, so by Lemma\ref{lem:norm_forms_are_complete_up_to_depth}, if $\varphi$ is equivalent to a normal form, inthe first case it is complete and in the second case it is complete up to its depth. 
\end{proof}

Therefore, all modal logics have formulas that are complete up to their depth. 
In fact, for any finite set of propositional variables $P$ and $d \geq 0$, we can define $\varphi_P^d = \bigwedge_{i=0}^d \Box^i \bigwedge P  $, which is equivalent in \t\ and \d\ to a normal form (by induction on $d$). Then, we can use a reduction similar to the one from the proof of Theorem \ref{thm:hardness} to prove that for every modal logic, completeness up to depth is as hard as provability.

\begin{proposition}\label{prp:up_to_depth_lower}
	For any complexity class $C$ and logic $l$,
	if $l$-provability is $C$-hard, then completeness up to depth is $C$-hard.
\end{proposition}

\begin{proof}
	The proof is similar to the one for Theorem \ref{thm:hardness} and is by reduction from $l$-provability. We are given a formula $\varphi \in L(P)$ --- and we assume that $P \neq \emptyset$. For $l \neq \t, \d$,  $\varphi_P^l(d) = \varphi_P^l$ as defined in Section \ref{sec:completeness}; for $l = \t$ or \d, let $\varphi_P^l(d) = \varphi_P^d$ as defined above. We also assume an appropriate $M_l,a_l \models \varphi_P^l(d)$. If $M_l,a_l \not\models \varphi$, let $\varphi_c = \bigwedge P \wedge \Box \top$; otherwise, let $\varphi_c = \varphi \rightarrow \varphi_P^l(d)$.
	For the second case, if $\varphi$ is provable, then $\varphi_c$ is equivalent to $\varphi_P^l(d)$, which is complete up to its depth. If $\varphi_c$ is complete up to its depth, then by Lemma \ref{lem:deep_completeness_is_character}, it is equivalent to a normal form $\psi \in F_P^d$. So, $\psi$ is equivalent to $\varphi_c = \varphi \rightarrow \varphi_P^l(d)$, which is equivalent to $\neg \bigvee S \vee \varphi_P^l(d)$ for some $S \subseteq F_P^d$, by Theorem \ref{thm:fine}. Since normal forms are mutually exclusive, $\bigvee S$ is equivalent to $\neg \bigvee (F_P^d \setminus S)$, so $\psi$ is equivalent to $\bigvee (F_P^d \setminus S) \vee \varphi_P^l(d)$. Therefore, either $S = F_P^d$ and $\psi$ is equivalent to  $\varphi_P^l(d)$, or $F_P^d \setminus S$ is a singleton of a normal form equivalent to $\varphi_P^l(d)$.
	In the first case, $\varphi$ is  provable, because for any model $\M,a$, by Theorem \ref{thm:fine}, $\M,a \models \bigvee F_P^d$, so $\M,a \models \varphi$.
	The second case cannot hold, because it would mean that $\varphi$ is equivalent to $\neg \varphi_P^l(d)$, but $M_l,a_l \models \varphi$.
\end{proof}
We demonstrate that this variation of the completeness problem is in \PSPACE\ when the logic is \k; it seems plausible that one can follow similar approaches that use normal forms for the remaining modal logics.
\begin{proposition}\label{prp:k_up_to_depth}
	A formula $\varphi$ is complete up to its depth for \k\ if and only if $\varphi \wedge \Box^{md(\varphi)+1}\bot$ is complete for \k.
\end{proposition}
\begin{proof}
Let $\psi \in F_P^d$ be a normal form. Then, $\psi \wedge \Box^{d+1}\bot$ is equivalent in \k\ to $\psi^{+1} \in F_P^{d+1}$, which is $\psi$ after we replace all $\Diamond \psi'$ in $\psi$  by $\Diamond(\psi' \wedge \Box \bot)$, where $\psi' \in F_P^0$. 
Notice that $\psi_1, \psi_2 \in F_P^d$ are distinct normal forms if and only if $\psi_1^{+1}, \psi_2^{+1}$ are distinct normal forms in $F_P^r$ for every $r>d$.
So, $\varphi$ is complete up to its depth for \k\ if and only if $\varphi \wedge \Box^{md(\varphi)+1}\bot$ is complete for \k. 
%
\end{proof}
%
%

\subsection{More Logics}

There is more to \ML, so perhaps there is also more to discover about the completeness problem. 
We based the decision procedure for the completeness problem for each logic on a decision procedure for satisfiability. We distinguished two cases:
\begin{itemize}
	\item If the logic has axiom $5$, then to test satisfiability we guess a small model and we use model checking to verify that the model satisfies the formula. This procedure uses the small model property of these logics (Corollary \ref{cor:smallNImodels}).
	To test for completeness, we guess \emph{two} small models; we verify that they satisfy the formula and that they are non-bisimilar. 
	We could try to use a similar approach for another logic based on a decision procedure for satisfiability based on a small model property (for, perhaps, another meaning for ``small'').
	To do so successfully, a small model property may not suffice. We need to first demonstrate that for this logic, a formula that is satisfiable and incomplete has \emph{two} small non-bisimilar models.
	
	\item For the other logics, we can use a tableau to test for satisfiability. We were able to combine the tableaux for these logics with bisimulation-testing to provide an optimal --- when the completeness problem is not trivial --- procedure for testing for completeness. For logics where a tableau gives an optimal procedure for testing for satisfiability, this is, perhaps, a promising approach to also test for completeness.
\end{itemize}
%
Another direction of interest would be to consider axiom schemes as part of the input --- as we have seen, axiom $5$ together with $\varphi^\sv$ is complete for \t, when no modal formula is.
%

\subparagraph*{Acknowledgments.}

The author is grateful to Luca Aceto 
for valuable comments that helped improve the quality of this paper.

\bibliography{mybib}

\end{document}